\documentclass{article}
	
\usepackage[letterpaper,margin=1.5in]{geometry}
\usepackage{amsmath, amssymb, amsthm, amsfonts,bm}
\usepackage[inline]{enumitem}
\usepackage{csquotes}
\usepackage{hhline}

\usepackage{soul}
\usepackage{cite}
\usepackage{framed}
\usepackage[framemethod=tikz]{mdframed}
\usepackage{appendix}
\usepackage{graphicx}
\usepackage{color}
\usepackage{wrapfig}
\usepackage{multirow}
\usepackage{algorithm, subcaption}
\usepackage[noend]{algpseudocode}
\usepackage[textsize=tiny]{todonotes}
\usepackage{enumitem}
\usepackage{xspace}
\definecolor{darkgreen}{rgb}{0,0.5,0}
\definecolor{darkblue}{rgb}{0,0,0.8}
\usepackage{hyperref}
\hypersetup{
    unicode=false,          % non-Latin characters in Acrobat’s bookmarks
    colorlinks=true,        % false: boxed links; true: colored links
    linkcolor=darkblue,          % color of internal links (change box color with linkbordercolor)
    citecolor=darkgreen,        % color of links to bibliography
    filecolor=magenta,      % color of file links
    urlcolor=cyan           % color of external links
}
\RequirePackage[]{silence}
\usepackage[nameinlink,capitalize]{cleveref}

\theoremstyle{theorem}
\newtheorem{theorem}{Theorem}[section]
\theoremstyle{lemma}
\newtheorem{lemma}[theorem]{Lemma}
\theoremstyle{corollary}
\newtheorem{corollary}[theorem]{Corollary}
\theoremstyle{claim}

\theoremstyle{definition}
\newtheorem{definition}{Definition}[section]
\theoremstyle{remark}

\newcommand{\mixing}[1]{\ensuremath{\tau_{#1}}\xspace}
\newcommand{\mixingG}{\mixing{G}}

\newcommand{\Var}[1]{{\textup{V}}\left [ #1 \right ]}

\newcommand{\Expect}[1]{ \textup{E}\left [ #1 \right ]}

\newcommand{\ignore}[1]{}

\algnewcommand\algorithmicswitch{\textbf{switch}}
\algnewcommand\algorithmiccase{\textbf{case}}

\algdef{SE}[SWITCH]{Switch}{EndSwitch}[1]{\algorithmicswitch\ #1\ \algorithmicdo}{\algorithmicend\ \algorithmicswitch}%
\algdef{SE}[CASE]{Case}{EndCase}[1]{\algorithmiccase\ #1}{\algorithmicend\ \algorithmiccase}%
\algtext*{EndSwitch}%
\algtext*{EndCase}%
\newcommand{\eps}{\epsilon}

\newcommand{\congest}{\ensuremath{\mathsf{CONGEST}\xspace}}
\newcommand{\gossip}{\ensuremath{\mathsf{GOSSIP}\xspace}}

\newcommand{\poly}{\operatorname{\text{{\rm poly}}}}
\newcommand{\polylog}{\operatorname{\text{{\rm polylog}}}}
\newcommand{\floor}[1]{\lfloor #1 \rfloor}

\newcommand{\Prob}[1]{\Pr\left(#1\right)}

\newcommand{\ID}{\operatorname{ID}}

\newcommand{\depth}{\operatorname{\text{\tt depth}}}
\newcommand{\roundup}[1]{ \lceil #1 \rceil}

\newcommand{\size}[1]{\left| #1 \right|}
\newcommand{\indicator}[1]{I\left[ #1 \right]}
\newcommand{\val}[1]{\operatorname{val}\left( #1 \right)}
\newcommand{\head}[1]{\operatorname{head}\left( #1 \right)}
\newcommand{\tail}[1]{\operatorname{tail}\left( #1 \right)}

\renewcommand{\paragraph}[1]{\vspace{0.15cm}\noindent {\bf #1}:}

\begin{document}

\date{}

\title{Distributed Data Summarization in Well-Connected Networks} 

\author{
Hsin-Hao Su\\
	 Boston College \\
	 suhx@bc.edu 
\and
Hoa T.~Vu\\
	 Boston College \\
	 vuhd@bc.edu
}

\maketitle

%!TEX root = main.tex

\pagenumbering{gobble}

%\vspace{-15mm}
\begin{abstract}
\normalsize
We study distributed algorithms for some fundamental problems in data summarization. Given a communication graph $G$ of $n$ nodes each of which may hold a value initially, we focus on computing $\sum_{i=1}^N g(f_i)$, where $f_i$ is the number of occurrences of value $i$ and $g$ is some fixed function. This includes important statistics such as the number of distinct elements, frequency moments, and the empirical entropy of the data. 

In the $\congest$ model, a simple adaptation from streaming lower bounds shows that it requires $\tilde{\Omega}(D+ n)$ rounds, where $D$ is the diameter of the graph, to compute some of these statistics exactly. However, these lower bounds do not hold for graphs that are well-connected. We give an algorithm that computes $\sum_{i=1}^{N} g(f_i)$  exactly in $\mixingG \cdot 2^{O(\sqrt{\log n})}$ rounds where $\mixingG$ is the mixing time of $G$. This also has applications in computing the top $k$ most frequent elements.

 We demonstrate that there is a high similarity between the $\gossip$ model and the $\congest$ model in well-connected graphs. In particular, we show that each round of the $\gossip$ model can be simulated almost perfectly in $\tilde{O}(\mixingG)$ rounds of the $\congest$ model.  To this end,  we develop a new algorithm for the $\gossip$ model that $1\pm \eps$ approximates the $p$-th frequency moment $F_p = \sum_{i=1}^N f_i^p$ in $\tilde{O}(\eps^{-2} n^{1-k/p})$ rounds \footnote{$\tilde{O}$ omits $\polylog(n)$ factors.}, for $p \geq 2$, when the number of distinct elements $F_0$ is at most $O\left(n^{1/(k-1)}\right)$. This result can be translated back to the $\congest$ model with a factor $\tilde{O}(\mixingG)$ blow-up in the number of rounds.
 \end{abstract}
 \pagebreak
 \pagenumbering{arabic}
 \setcounter{page}{1}

 % \thispagestyle{empty}

%\newpage
%\tableofcontents
%\newpage

\section{Introduction}

\paragraph{Motivation} Analyzing massive datasets has become an increasingly important and challenging problem. Collecting the entire data to one single machine is usually infeasible due to  memory, I/O, or network bandwidth constraints. Furthermore, in many cases,  data are distributed over the network and we hope to  aggregate some of their properties efficiently.  In this work, we consider several fundamental  data summarization problems in distributed networks, specifically in the $\congest$ and $\gossip$ models.

In this problem, we have a graph $G=(V,E)$ of $n$ nodes. Each node $v$ in the graph may hold a value $\val{v}$ in the range $ \{ 1,\ldots,N\} \cup \{ \textup{NULL} \}$ where $\textup{NULL}$ simply means that the node does not hold a value. If $\val{v}= \textup{NULL}$, we call $v$ an {\em empty node}. 

We often use the notation $[N] := \{1,\ldots,N\}$. Let $f_i$ be the number of  nodes that hold value $i$, i.e., $f_i = \size{\{ v \in V : \val{v} = i \}}$. We want to compute  $
\sum_{i=1}^N g(f_i)$  for some fixed function $g$. To demonstrate some important cases, consider the following examples.

Consider $g(f_i) = 1$ if $f_i > 0$ and 0 otherwise. This corresponds to the problem of counting the number of distinct elements (or computing the 0-th frequency moment $F_0$). The problem may arise in the following situation: Each node stores a version of a file (e.g.~the hash of a blockchain), and we want to know how many different versions there are in the network. 

If $g(f_i) = f_i^p$ for some fixed $p =2,3,\ldots$, then this corresponds to the problem of computing the $p$-th frequency moment $F_p$. We note that $F_p$ is a basic, yet very important
statistic of a dataset. $F_2$ measures the variance and could be used to estimate the size of a self-join in database applications. For higher $p$, $F_p$ measures the skewness  of the dataset (see \cite{AlonMS99}). Note that $F_1$ can be computed in $O(D)$ rounds by aggregating along a breath-first-search (BFS) tree.

Another example is $g(f_i) = - (f_i/F_1) \cdot \log (f_i/F_1)$.  In this case, the sum is the empirical entropy of the data.  Computing the empirical entropy is motivated by network applications such as detecting anomalies \cite{GuMT05,XuZB05,WagnerP05}.

%In addition to the above, we also consider the related problem of computing the mode, i.e., the most frequent element in the network. A more general question would be identifying the top $k$ frequent elements. We note that the number of occurrences of the mode is also denoted by $F_\infty$.

\paragraph{Models} We now give  a formal description of the $\congest$ and $\gossip$ models, where the running time of an algorithm is measured by the number of rounds. %In these two models, the network operates in synchronous rounds. In each round,  a node may communicate with some other nodes. In the $\congest$ model, each node communicates with its neighbors. In the $\gossip$ model, the underlying graph is a complete graph, each node picks random node to communicate with. In both models, the message size per link per round is restricted to $O(\log n)$. The running time of an algorithm is measured as the number of rounds. We formally define the models below.

 \begin{definition}
In the $\congest$ model, we are given a graph $G = (V,E)$ of $n$ nodes,  in each synchronous round, each node can talk (send and receive message) to each of its neighbors and then perform local computations. Each message is restricted to be at most $O(\log n)$ bits. 
\end{definition}
 
 \begin{definition}
In the $\mathsf{GOSSIP}(\lambda)$ model with $n$ nodes, in each synchronous round, each node $u$ samples a node $t(u)$ from a distribution that satisfies the following: For any node $v$ and any subset of nodes $Z$ where $u \notin Z$,
$$\Pr\left(t(u) = v \middle\vert \bigwedge_{z \in Z} t(z) \right)  \in \left[\frac{1-\lambda}{n}, \frac{1+\lambda}{n} \right].$$
In the above, ``$\bigwedge_{z \in Z} t(z)$'' means conditioning on any assignment of each $t(z)$ for $z \in Z$. Then, $u$ can  PUSH a message of size $O(\log n)$ to $t(u)$ or PULL a message of size $O(\log n)$ from $t(u)$. Then, after performing some local computations, it proceeds to the next round. We refer $\mathsf{GOSSIP}$ model as the $\mathsf{GOSSIP}(0)$ model.
\end{definition}

\subsection{Our results} 

We organize our main results into three categories: a) results in the $\congest$ model, b) an emulation of the $\gossip$ model in the $\congest$ model, and c) results in the $\gossip$ model. 

\paragraph{Results in the $\congest$ model} We briefly show how to adapt streaming algorithms to approximate $F_p$ (for $p = 0,2,3,\ldots)$ in the $\congest$ model. We also demonstrate some  lower bounds and conditional lower bounds that give evidence that such algorithms are optimal or near-optimal.

  The lower bounds show that computing $F_p$ exactly for $p =0,2,3, \ldots$ requires $\tilde{\Omega}(D + n)$ rounds and approximating $F_p$ within a constant factor requires {\it polynomial} rounds in $n$ for $p \geq 3$. Roughly speaking, the hard instances in the $\congest$ model are graphs with a small balanced cut of $O(1)$ size that causes an information bottleneck. However,  such bottleneck does not occur in graphs that are well-connected. Our first main result aims to answer the following question: {\em Could one design more efficient algorithms for well-connected graphs?} We give a positive answer to this question.

By using the permutation routing algorithms of Ghaffari et al.~\cite{GhaffariKS17} (later improved by Ghaffari and Li \cite{GhaffariL18}), we show that there exists an algorithm running in $\mixingG \cdot 2 ^{O(\sqrt{ \log n})}$ rounds that computes $\sum_{i=1}^N g(f_i)$ for all fixed and computable functions $g$ with high probability (w.h.p.) \footnote{We consider $1-1/\poly(n)$ as high probability. }. This includes all the aforementioned  quantities such as the number of distinct elements, higher frequency moments, and the empirical entropy. Thus, if the graph has small mixing time such as expanders \cite{Goldreich11o, hoory2006expander}, where $\mixingG = \polylog(n)$, then we obtain a much more efficient sub-polynomial in $n$ algorithm compared to the adaptation of the streaming counterpart.

\begin{theorem}[{\em Main result 1}] \label{thm:main1}
There exists an algorithm that computes $\sum_{i=1}^N g(f_i)$ exactly for all (fixed and computable) functions $g$   in the $\congest$ model in $\mixingG \cdot 2^{O(\sqrt{\log n})} $ rounds w.h.p. 
\end{theorem}

Our algorithm can also easily be extended to find the top  $k$ frequent elements in $O(k)+\mixingG \cdot 2 ^{O(\sqrt{ \log n})}$ rounds. 

%We also extend our algorithm to the following:
% \begin{itemize}[leftmargin=*]
%\item Find the top  $k$ frequent elements in $O(k)+\mixingG \cdot 2 ^{O(\sqrt{ \log n})}$ rounds.
%
%
%\item Select the $k$-th rank element after {\em removing duplicates} in $\mixingG \cdot 2 ^{O(\sqrt{ \log n})}$ rounds.
%\end{itemize}

%Our algorithm makes use of the routing algorithms by Ghaffari et al. \cite{GhaffariKS17} and Ghaffari-Li \cite{GhaffariL18}.

\paragraph{From $\congest$ to $\gossip$} The  lower bounds do not apply directly to the $\gossip$ model either. This is because for any balanced cut of the nodes, one expects $O(n)$ messages to be sent across in one round. Moreover, the expected communication degree per node in the $\gossip$ model is $O(1)$. Intuitively, the graph formed by the communication pattern in the $\gossip$ model is similar to an expander graph.

 In fact, we show that well-connected graphs can emulate  the $\gossip$ model efficiently. In particular, one round of the $\gossip(1/\poly(n))$ model can be emulated in $\mixingG \cdot \polylog(n)$ rounds in the $\congest$ model where the underlying graph is $G$.  Therefore, any algorithm that works in the $\gossip(1/\poly(n))$ model can be turned into an algorithm in the $\congest$ model with an $\tilde{O}(\mixingG)$ factor blow-up. 

Consider our results in the $\congest$ model. The permutation routing algorithms of \cite{GhaffariKS17} and \cite{GhaffariL18} introduces a super-logarithmic factor, $2^{O(\sqrt{\log n})}$, on top of the mixing time. It becomes the bottleneck in graphs with small mixing times (e.g.,~expanders). Improving the permutation routing algorithm directly yields improvements to our results in the $\congest$ model (and many other problems). However, it is unclear if it can be improved.
This emulation result serves as an alternative route to circumvent the $2^{O(\sqrt{\log n})}$ factor, if one develops efficient $\gossip$ algorithms.

%	The emulation result has an implication: The permutation routing algorithms of \cite{GhaffariKS17} and \cite{GhaffariL18} introduces a super-logarithmic factor, $2^{O(\sqrt{\log n})}$, on top of the mixing time. It becomes the bottleneck in graphs with small mixing times (e.g.~expanders). It is yet to be investigated if this factor can be improved, say, to $\polylog(n)$. If this were indeed the case, then our algorithms can also be improved to $\tilde{O}(\mixingG)$ time. The emulation result serves an alternative route to circumvent the $2^{O(\sqrt{\log n})}$ factor, if one develops efficient algorithms in the $\gossip$ model.

%We take an alternative route to circumvent the $2^{O(\sqrt{\log n})}$ factor: We show that well-connected graphs can emulate a well-studied model, the $\gossip$ model, efficiently. In particular, one round of the $\gossip(1/\poly(n))$ model can be emulated in $\mixingG \cdot \polylog(n)$ rounds in the $\congest$ model where the underlying graph is $G$.  Therefore, any algorithms that works in the $\gossip(1/\poly(n))$ model can be turned into an algorithm in the $\congest$ model with a $\tilde{O}(\mixingG)$ factor blow-up. 

\begin{theorem}[{\em Main result 2}]\label{thm:gossipemu}
For $\lambda = 1/\poly(n)$, one round of the $\mathsf{GOSSIP}(\lambda)$ model can be emulated in $\tilde{O}(\tau_G)$ rounds in the $\mathsf{CONGEST}$ model where $G$ is a connected graph denoting the underlying network.\end{theorem}

We believe that this emulation result may be of independent interest. 
Jelasity et al.~\cite{Jelasity2007} studied how to implement the gossip-based peer sampling service empirically. Our result is an additional way to implement the service with theoretical guarantees.

%Motivated by our algorithms, we show that one round of the \textsf{GOSSIP} model can be emulated almost-perfectly in $\tilde{O}(\tau_G)$ rounds in the \textsf{CONGEST} model. While this does not lead to any direct improvement for the known algorithms, suppose one could improve the running time of the routing algorithm by  Ghaffari et al. \cite{GhaffariL18}, Ghaffari-Li \cite{GhaffariKS17} to just $\mixingG \cdot \polylog(n)$ rounds, then we have a more efficient algorithm to approximate higher frequency moments when $F_0$ and $\mixingG$ are small. The main idea is that one can  simulate one round of the $\gossip$ model in the $\congest$ model using random walks. The main challenge that we must overcome is that for non-regular graphs, the stationary distribution is not uniform as in the $\gossip$ model. 

\paragraph{Results in the $\gossip$ model} Motivated by our emulation result, we develop algorithms for the $\gossip$ model. In particular, we are interested in the following question: {\em Suppose the number of non-empty nodes are sublinear in $n$. Could we take advantage of the computational power of the empty nodes?}

Suppose that  the number of non-empty nodes is at most $O(n^{1/(k-1)})$ (or more generally, $F_0 \leq O(n^{1/(k-1)})$). We show that for any $p \geq 2$, $F_p$ can be approximated within a $1 \pm \eps$ factor in $O(\eps^{-2} n^{1-k/p} \log^2 n )$ rounds with high probability.

\begin{theorem}[{\em Main result 3}]\label{thm:fp-gossip}
If $F_0 = O(n^{1/(k-1)})$ for some integer $2 \leq k \leq p$, then there exists an algorithm that approximates $F_p$ up to a $1\pm \eps$ factor in $O(\eps^{-2} n^{1-k/p}\log^2 n)$ rounds in the $\gossip(1/n^c)$  model, for some sufficiently large constant $c$, w.h.p.
\end{theorem} 

The $\gossip(1/n^c)$ model will incur a $\pm 1/\poly(n)$ additive error which we consider insignificant. Since $F_0\leq n$,  we have an algorithm that approximates $F_2$ in $\tilde{O}(\epsilon^{-2})$ rounds by setting $k=2$.  When $k > 2$, the empty nodes serve as the extra computation power to solve the problem. In such scenarios, we are able to obtain running time that is not known to be achievable by adapting the streaming counterpart. For example, when $k=3$, $F_0 = O(n^{1/2})$, we may approximate $F_3$  within a constant factor in $\polylog(n)$ rounds. Direct adaption of known streaming algorithms \cite{AndoniKO11,MonemizadehW10,AlonMS99} requires super-logarithmic rounds, even in the case where $F_0 = O(n^{1/2})$.

Combining Theorem \ref{thm:fp-gossip} and Theorem \ref{thm:gossipemu} with $k=p$, we  have the following corollary.
\begin{corollary}\label{cor:main4}
If $F_0 = O(n^{1/(p-1)})$, then there exists an algorithm in the $\congest$ model that approximates $F_p$ up to a $1\pm \eps$ factor in $\tilde{O}(\eps^{-2} \cdot \mixingG)$ rounds w.h.p.
\end{corollary}

%Suppose that  there are at most $O(n^{1/(k-1)})$ non-empty nodes, our results for the $\gossip$ model are as follows.\begin{itemize}[leftmargin=*]
%\item We give a perfect $\ell_k$-sampling algorithm that  samples an element $R$ such that $\Prob{R = i} = f_i^k/F_k$ in $O(\log n)$ rounds. Furthermore, we can also approximate $F_k$ up to a $1 \pm \eps$ factor in $O(\eps^{-2} \log n)$ rounds with high probability. 
%\item Based on the above primitives, we can approximate $F_p$ up to a $1 \pm \eps$ factor using $O(\eps^{-2} n^{1-k/p} \polylog(n)) $ rounds with high probability.
%\end{itemize}

%Furthermore, our algorithm saves a factor $\log n$ for approximating $F_2$ with high probability compared to adapting the tug-of-war sketch  which uses $O(\eps^{-2} \log^2 n)$ rounds ($O(\eps^{-2} \log^3 n)$??).

\begin{figure}\label{tbl:result}
\centering
\begin{tabular}{|c|c|c| c|}
\hline
& & &\\
[-1em]
& Number of rounds & Assumption & Approximation \\
[-1em]
 &   & & \\
 \hline 
 & &  &\\
 [-1em]
\multirow{2}{*}{$\congest$} &  $\mixingG \cdot 2^{O(\sqrt{\log n})}  $ (*)&  & Exact  \\\cline{2-4}
& & &\\
 [-1em]
 &    $O(\eps^{-2} \mixingG \cdot \polylog n)$ & $F_0 \leq O(n^{1/(p-1)})$ & $1 \pm \epsilon$ \\
\hline
 & &  & \\
 [-1em]
$\gossip$  & $O(\eps^{-2} \cdot n^{1-k/p} \cdot \polylog n)$ & $F_0 \leq O(n^{1/(k-1)})$ & $1 \pm \epsilon$\\
 \hline
\end{tabular}
\caption{Results summary for computing frequency moments $F_p$. (*)  can also be used to compute $\sum_i g(f_i)$ for all fixed and computable functions $g$.}
\end{figure}

%To obtain the algorithm, we gave an algorithm for approximating $F_k$ and a perfect $\ell_k$-sampling primitive that run in $\polylog(n)$ rounds when $F_0 = O(n^{1/(k-1)})$ and the number of non-empty nodes is $\Omega(n)$.  By using the $\ell_k$-sampling primitive along with the estimated $F_k$ value, we show how to estimate $F_p$ in $\tilde{O}(n^{1-k/p})$ rounds. We also how to reduce the scenario where there are $O(n^{1/(k-1)})$ nodes holding values to the scenario where $F_0 = O(n^{1/(k-1)})$ and the number of non-empty nodes is $\Omega(n)$. This is done by duplicating all the input values exactly the same time.

%Based on our result, we observe that the smaller $F_0$ is (or equivalently, when $N$ is small), the more efficient algorithm we obtain. This observation has a nice corollary: suppose only $O(n^{1/(k-1)})$ nodes hold values in the beginning, then we will show that the extra nodes (i.e., the  empty nodes)  provide extra computational power to solve the problem. For this to work, we need to design an extra algorithm that duplicates the values such that at least $\Omega(n)$ nodes holds values in $[N]$ and the ratios among the frequencies are preserved.
% To the best of our knowledge, no similar result based on the parameterization of $F_0$ is known in the streaming algorithms literature.

\subsection{Related work and preliminaries} \label{sec:prelim} 
\paragraph{Related work} In the distributed setting, Kuhn et al. \cite{KuhnLS08} studied the problem of finding the mode, i.e., the most frequent element, in the $\congest$ model. Let $D$ is the diameter of the graph, and $f^*$ is the largest number of occurrences among the values.  They gave an algorithm that uses $O(D+ F_2/f^* \cdot \log k)$ rounds. They also briefly explained how to implement streaming algorithms for approximating $F_0$ and $F_2$ in the $\congest$ models. Also related to data summarization, Kuhn et al. \cite{KLW07} designed selection algorithms in the $\congest$ model. 

In the data stream model, each stream token $(i,x)$ corresponds to the update $f_i \leftarrow f_i + x$. The problem of approximating the number of distinct elements $F_0$ and frequency moments $F_p$ have been extensively studied. An incomplete list includes \cite{Bar-YossefJKST02,GibbonsT01,KaneNW10,AlonMS99,IndykW05,Indyk06,AndoniKO11,Woodruff04,Ganguly15}. Roughly speaking, the space complexity for approximating $F_p$ in the data stream model is $\tilde{O}(\eps^{-2})$  for $0 \leq p \leq 2$ and $\tilde{O}(\eps^{-2}n^{1-2/p})$ for $p \geq 2$. Furthermore, it is known that approximating $F_\infty$ (or identifying the mode) is not possible in sublinear space. In the data stream model, researchers have also studied the problem of approximating the entropy \cite{HarveyNO08,ChakrabartiBM06,ChakrabartiCM10}. 

We will briefly discuss the similarities between the data stream model and the $\congest$ model. Roughly speaking, since streaming algorithms use little memory, they can be adapted to the $\congest$ model by passing the memory state of the corresponding algorithm along the breadth-first-search tree.  Similarly, lower bounds from streaming algorithms literature can also be translated into lower bounds in the $\congest$ model. Data aggregation problems have also been studied in  directed networks \cite{KuhnO11}.

There is also a rich literature in the $\gossip$ model started by the work of \cite{DGHILSSSD87}.  Some examples include spreading message \cite{FG85,Pittel87,KSSV00}, computing the sum and average \cite{kempe2003gossip, CP12,KDNRS06}, renaming \cite{Giakkoupis13}, and quantile computation \cite{HaeuplerMS18}.

%\subsection{Preliminaries}\label{sec:prelim} Let $D =\diameter(G)$. 

\paragraph{Preliminaries} We introduce basic notations and algorithmic building blocks in the \congest model. %See Appendix \ref{appendix:primitives}.

To ease our presentation, we assume $N = O(\poly(n))$. In our algorithms, we often want to learn about the sum of all the values (or hash values, indicator variables) held by the nodes; this can be done in $O(D)$ rounds. Another algorithmic primitive, based on downcasts and upcasts, is to broadcast the $k$ smallest values in $O(D + k)$ rounds. 

%We also use standard primitives such as 

%On the BFS tree, we can aggregate the values from the leaves upward so that the root learns about the sum. Finally, the root downcasts the sum back to the rest of the graph. These steps take $O(D)$ rounds.

%On the BFS tree, in each round, a node upcasts the smallest value in its memory to its parent if that value has not been upcasted before. After $O(D + k)$ rounds, the root will have the $k$ smallest values and downcast them back to other nodes (see Lemma 4.3.1 in \cite{Peleg:2000}).

We define the mixing time similarly to \cite{GhaffariKS17}.  A lazy random walk is a random walk in which at each step, we stay at the same node with probability 0.5 and move to a random neighbor with probability 0.5. Lazy random walk ensures  the existence of a unique stationary distribution (i.e.,  the walk is aperiodic). From now  on, we simply refer to a lazy random walk as a random walk. 

Let $P^t_u=(P^t_u(v_1), \ldots, P^t_u(v_n)) \in [0,1]^{n}$ denotes the probability distribution on the nodes after $t$ steps of a lazy random walk that starts at $u$. A crucial property of a random walk is that it will converge to the stationary distribution $(\deg(v_1) / 2m , \ldots, \deg(v_n) / 2m)$. Define the mixing time $\tau_G$ to be the minimum $t$ such that for any starting node $u$ and any node $v_i$,
$$\left| P^t_u(v_i) - \frac{\deg(v_i)}{2m} \right| \leq \frac{\deg(v_i)}{2mn}~.$$

Using an $O(D)$-round pre-proscessing, we can  assume that each node has a unique  $\ID$ in $[n]$.  Suppose we want the nodes in a graph to have unique IDs in $[n]$. We can elect a leader and build a breadth-first-search (BFS) tree that is rooted at the leader  in $O(D)$ rounds \cite{Peleg:2000}. Each node $u$ can learn about the number of nodes in $T_v$ where $v$ is a child of $u$ and $T_v$ is the subtree that is rooted at $v$. This is done by aggregating the size from the leaves upward. It is then straightforward to assign the $\ID$s to the nodes based on the depth-first-search (DFS) ordering. Specifically, the root notifies each of its children $v$ the range of the $\ID$s in $T_v$, based on the DFS ordering, and then recurse on $T_v$.  From now on, we can refer to the nodes by their $\ID$s, i.e., $\ID(v) = v$. 

We will also make use of hash functions. An $O(1)$-wise independent hash function $h: [a] \rightarrow [b] $ where $a$ and $b$ are at most $\poly(n)$  can be stored in $O(\log n)$ bits. Hence, if we need to use a hash function, a leader can broadcast such  hash function (using a BFS tree) in $O(D + \log n)$  rounds in the \congest model and $O(\log n)$ rounds in the \gossip model.  

%Suppose the nodes of $G$ do not have unique $IDs$. 

%heavy-hitter algorithms \cite{CharikarCF04,CormodeM05}.

%!TEX root = main.tex

\section{Algorithms in the CONGEST Model}\label{sec:exactcongest}

\subsection{Approximation algorithms}
\paragraph{Upper bounds} We show that we can adapt  the streaming algorithms given by Bar-Yossef et al. \cite{Bar-YossefJKST02} (for approximating $F_0$) and by Alon et al. \cite{AlonMS99} (for approximating $F_p$, where $p \geq2$) to the $\congest$ model (see Appendix  \ref{appendix:adaptation}).  This is not of particular novelty though we need some careful pipelining arguments to optimize the number of rounds. Kuhn et al. \cite{KuhnLS08} also briefly outlined similar results. However, the exact round-complexity for a good approximation w.h.p. is not very clear from their paper. 

%
%This subsection is a careful adaptation of the streaming algorithms given by Bar-Yossef et al. \cite{Bar-YossefJKST02} (for approximating $F_0$) and by Alon et al. \cite{AlonMS99} (for approximating $F_p$, where $p \geq2$) to the $\congest$ model. Therefore, this is not of particular novelty. However, we will later present lower bounds suggesting that these algorithms are tight. This provides  a clear context and a good motivation for the exact algorithm in the next subsection which is one of our main results. 

\begin{theorem} \label{thm:approximate-congest-f02}
There exists an $O(D + \eps^{-2} \log n)$-round algorithm in the $\congest$ model that computes a $1 \pm \eps$ approximation of $F_0$ and $F_2$ w.h.p. Furthermore, for $p > 2$, there exists an $O(D + \eps^{-2}\min(n, N)^{1-1/p} \log n)$-round algorithm  that computes a $1 \pm \eps$ approximation of $F_p$ w.h.p. 
\end{theorem}

%For simplicity, consider the setting where $n = N$.  

\paragraph{Lower bounds}  We show that the dependence on $\eps$ is tight via a conditional lower bound. Moreover, computing $F_p$ exactly requires $\tilde{\Omega}(n)$ rounds. The lower bounds are obtained by adapting the existing streaming lower bounds to the $\congest$ model. Due to space constraint, we refer to  Appendix  \ref{appendix:adaptation} for the discussion.

\begin{theorem}\label{thm:lower-bounds} 
We have the following lower bounds in the $\congest$ model.
 \begin{itemize}
\item If the conjecture in \cite{BrodyC09} holds, then approximating $F_p$ (for fixed $p \neq 1$) up to a $1 \pm \eps$ factor  requires $\Omega(D+\eps^{-2}/\log n)$ rounds.
\item A $(1 \pm 0.1)$-approximation of $F_p$, for $p > 2$,  requires $ \Omega\left(D+\left(N^{1-\frac{2}{p}}+n^{\frac{1-2/p}{1+1/p}}\right)/\log n\right)$ rounds.
\item Computing $F_p$ exactly requires $\Omega(D+n/\log n)$ rounds.
\end{itemize}
\end{theorem}

Hence, we cannot expect a sublinear algorithms (in terms of $N,n$) when $\eps \ll 1/\sqrt{n}$ or when we want to obtain the exact answer. The lower bounds arise in graphs with a small balanced cut which causes an information bottleneck. This observation motivates us to design an exact algorithm when the graph is well-connected.

\subsection{An exact algorithm in near mixing-time} \label{sec:exact-algorithm}

In this subsection, we show that it is possible to beat the  lower bounds and achieve an exact algorithm in sublinear time if the graph has fast mixing time.  For example, expander graphs are sparse and have $O(\polylog n)$ mixing time. 
 
% These  are sparse graphs, with constant maximum degree, with strong connectivity and very fast mixing time $O(\log n)$. For a thorough overview on expanders, we refer to \cite{Goldreich11o, hoory2006expander}. 

%For simplicity, we assume $m = n$. For general $m$, as long as $m \leq \poly(n)$, the round complexity only requires a $\polylog(m)$ factor overhead.

Suppose each node has a set of messages (of size $\polylog(n)$) each of which has a destination that is another node. In parts of our algorithms, we want to route  messages in a small number of rounds. We rely on the following routing algorithm in the $\congest$ model that uses $\mixingG \cdot 2^{O(\sqrt{\log n })}$ rounds.  We note that $2^{O(\sqrt{\log n })}$ is more than $\polylog n$ but smaller than any $n^{\eps}$ for $\eps > 0$.  Also  note that $D = O(\mixingG)$. Let $\deg(v)$ be the degree of $v$ in $G$.

\begin{theorem}[\cite{GhaffariL18}, \cite{GhaffariKS17}] \label{thm:routing}
If each node of $G$ is the source and the destination of at most $d_G(v) \cdot 2^{O(\sqrt{\log n })}$ messages, then there is a randomized algorithm in the \congest model that delivers all the  messages in $\mixingG \cdot 2^{O(\sqrt{\log n })}$ rounds  w.h.p.  %\footnote{The number of rounds in \cite{GhaffariKS17} is $\mixingG \cdot 2^{O(\sqrt{\log n \log \log n})} $ which was then improved to  $\mixingG \cdot 2^{O(\sqrt{\log n })} $ by \cite{GhaffariL18}.}
\end{theorem}

We also rely on the idea of sorting networks. Recall that we refer to the nodes by their unique $\ID$s in $[n]$. In a sorting network, in each step $r$, the sorting network will pick a set of disjoint pairs of nodes. We use $\val{x,r}$ to denote the  value that node $x$ holds in the beginning of step $r$. For each pair $x$ and $y$ (where $x < y$) that is picked, $x$ will keep the smaller value $\min(\val{x,r},\val{y,r})$ and $y$ will keep the larger value $\max(\val{x,r},\val{y,r})$. We treat $\textup{NULL}$ as $-\infty$. The sorting network can be constructed, solely based on $n$, so that after $t= O(\log n)$ steps, the values are sorted \cite{AjtaiKS83}. That is if $x < y$, then $\val{x,t} \leq \val{y,t}$. 

In the $\congest$ model, each node can generate the sorting network (note that the construction of the sorting network is independent of the topology of $G$ and the values held by the nodes). Furthermore, each step can be simulated by invoking Theorem \ref{thm:routing}. Thus, we have the following.
\begin{lemma}
In the $\congest$ model, we can sort the nodes' values in  $\mixingG \cdot 2^{O(\sqrt{\log n })} $ rounds w.h.p.
\end{lemma}

We now complete the proof of our first main result.

\begin{proof}[Proof of Theorem \ref{thm:main1}]
We now use $\val{v}$ to refer to the value that $v$ holds after sorting. We say a node $v$ is a {\em head} or a {\em tail} if $\val{v} \neq -\infty$ and its $\ID$ is the smallest or the largest respectively among the $\ID$s of the nodes that hold the value $\val{v}$. A node $v$ can tell that if it is a head or a tail by checking with the nodes $v+1$ and $v-1$ respectively using the routing algorithm in Theorem \ref{thm:routing}. We use $\head{i}$ and $\tail{i}$ to denote the $\ID$s of the head and the tail of value $i$ respectively.

Now, every node that is not a head or a tail marks its value as $-\infty$. Each remaining node  forms a token consisting of its value, $\ID$, and whether if it is a head or a tail (or both). We then use sorting networks again to sort the values in the graph. We will also swap the tokens if two nodes swap their values. Afterward, the head and the tail tokens of a value $i$ will be at some two nodes $v$ and $v+1$ (or just at a node $v$ if $f_i=1$). To this end, each node $v$ that holds a head token (that is not also a tail token) with value $i$ will check with  nodes $v+1$ and $v-1$, using the routing algorithm, to collect $\tail{i}$ since either $v+1$ or $v-1$ must have the tail token  of $i$. Now, $v$ can compute $g(f_i) =  g(\tail{i} - \head{i} +1 )$ and set this as its value. All the nodes that do not hold a head token set their values to 0. We then compute  $\sum_{i=1}^N g(f_i)$ using the BFS tree in $O(D)$ rounds. 
\end{proof}

The algorithm above is more robust compared to the AMS sketch since it can handle all fixed and computable functions $g$. The AMS sketch cannot guarantee sublinear space in the streaming model (or sublinear time in the $\congest$ model) for many functions \cite{BravermanO10a,BravermanC15,BravermanCWY16}.  The above algorithm also immediately leads to an algorithm that finds the top $k$ frequent elements.

\paragraph{Finding  the top $k$ frequent elements}  %Combining  with upcasts, we can easily find the most frequent element and its count $F_\infty = \max_i f_i$. Suppose we want to find the top $k$ frequent elements. After sorting, every node that is not a head or a tail marks its value as $-\infty$. Each remaining node forms a token with its value, $\ID$, and whether if it is a head or a tail. Now, we use sorting networks again to sort the values in the graph. We will also swap the token if two nodes swap their values. Afterward, each node $v$ that holds a head token will check with  node $v+1$ and $v-1$ (using the routing algorithm in Theorem \ref{thm:routing}) to compute the number of occurrences of the value in the token that it holds. We can then use upcasts to find the top $k$ frequent in $O(D+k)$ rounds. 
At the end of the above algorithm, the occurrence of each value $i$ is held by some node $v$. Recall we can find the top $k$ elements in the graph using $O(D+k)$ rounds via upcasts. This immediately leads to the following result.

\begin{theorem} \label{thm:topk}
There exists an algorithm that finds the top $k$ elements (along with their occurrences)  in the $\congest$ model in $O(k+\mixingG \cdot 2^{O(\sqrt{\log n})} )$ rounds w.h.p. 
\end{theorem}

%\paragraph{Selection without duplicates} We consider the selection problem of finding the $k$-th rank element {\em ignoring duplicates}. Previous work in both the $\congest$ and $\gossip$ models focused on the case with duplicates  \cite{kempe2003gossip,HaeuplerMS18}.  Let the set of values that occur in the graph be  $Z := \{ i \in [N] : f_i > 0 \}$. In particular, we want to find $j \in Z$ such that $\size{\{ i \in Z:  i < j \}} = k-1$. To do this, we follow the approach in the proof of Theorem \ref{thm:main1}, at the end, each node that hold a tail token (that is not also a head token) mark their values as $-\infty$. We then sort the values again and broadcast the value of node $k$.
%
%
%\begin{theorem} There exists an algorithm that computes the $k$-th rank element of the set $\{ i \in [N]: f_i > 0 \}$ exactly in the $\congest$ model in $\mixingG \cdot 2^{O(\sqrt{\log n})} $ rounds w.h.p. 
%\end{theorem}

%!TEX root = main.tex

\section{Emulation of GOSSIP Model in the CONGEST Model}

In Section \ref{sec:exactcongest}, we have shown that the moments can be computed exactly in $\tau_G \cdot 2^{O(\sqrt{\log n})}$ rounds. If the permutation routing algorithm can be improved to $\polylog(n)$ rounds, then the running time of our algorithms would be improved to $\tilde{O}(\tau_G)$ rounds. Whether the $2^{O(\sqrt{\log n})}$ factor can be improved to $\polylog(n)$ is an intriguing open question.

Instead of tackling the complexity of permutation routing, in this section, we show that one round of the \textsf{GOSSIP} model can be emulated almost-perfectly in $\tilde{O}(\tau_G)$ rounds in the \textsf{CONGEST} model. Therefore, if there is a $\polylog(n)$-round algorithm in the \textsf{GOSSIP} model, it implies a $\tilde{O}(\tau_G)$ rounds algorithm in the \textsf{CONGEST} model. In Section \ref{sec:gossip}, we present efficient algorithms in the $\gossip$ model when $F_0$ is small (or when the number of empty nodes is large) which can be translated back to the \textsf{CONGEST} model using the emulation result in this section.

Recall that $P^t_u=(P^t_u(v_1), \ldots, P^t_u(v_n)) \in [0,1]^{n}$ denotes the probability distribution on the nodes after $t$ steps of a lazy random walk that starts at $u$ (see Section \ref{sec:prelim}). Given $\lambda$, we let $\tau_G(\lambda)$ be the smallest $t$ such that for any starting node $u$ and any node $v_i$,
$$\left\vert P^t_u(v_i) - \frac{\deg(v_i)}{2m} \right\vert \leq \lambda.$$ Note that if $\lambda = 1/\poly(n)$ then $\tau_G(\lambda) = O(\tau_G)$ \cite[Definition 2.1]{GhaffariKS17}.

We will run several random walks in parallel. The following lemma from \cite{GhaffariKS17} shows that the parallel random walks can be performed efficiently in the \textsf{CONGEST} model.

\begin{lemma}[\cite{GhaffariKS17}, Lemma 2.5]\label{lem:parallelrandomwalk}Let $G = (V,E)$ be an $n$-node graph and let $t \geq 1$ be a positive integer. Assume that we perform $T = O(\poly(n))$ steps of a collection of independent random walks in parallel. If each node $u \in V$ is the starting node of at most $t \cdot \deg(u)$ random walks, w.h.p., the $T$ steps of all the random walks can be performed in $O((t + \log n) \cdot T)$ rounds in the \congest model. \end{lemma}

The main technical difficulty of the emulation lies in the fact that the stationary distribution is not necessarily uniform in general graphs. If $G$ is regular, we could let each node $u$ start a random walk that runs for $O(\tau_G)$ steps. The probability that $u$ ends at each node is (nearly) uniform. If it ends at $v$ then we set $t(u) = v$. Moreover, by Lemma \ref{lem:parallelrandomwalk}, all the random walks can be performed simultaneously in $\tilde{O}(\tau_G)$ rounds. 

In irregular graphs, such  approach does not work because the stationary distribution is not uniform. One remedy is to regularize the random walk (i.e.~adding self-loops to non-maximum degree nodes). However, this may significantly increase the mixing time of the graph (e.g., a star graph). In the following, we give an emulation algorithm whose running time is within a $\polylog(n)$ factor of the mixing time.
 
For each node $u$ in $G$, we split it into $\deg(u)$ {\it compartments}. When a random walk enters a node, it is assigned randomly to one of its compartments. There are $2m$ compartments in  $G$ in total. We outline the emulation algorithm below.

\begin{enumerate}

\item\label{step:1} Let $k = \lfloor 1.5m / n \rfloor$.  Each node creates $k$ {\it destination tokens} and distributes them over the compartments in $G$ so that each compartment contains at most {\it one} destination token. Now $n \cdot k \approx 1.5m$ compartments are filled with tokens. 

\item\label{step:2} Each node sends out a {\it source token}. Each source token starts a random walk to distribute itself randomly over the compartments at the end. If the source token of node $u$ ends in a compartment with the destination token of some node $v$, we set $t(u) = v$. 

\item\label{step:3} Route the message between $u$ and $t(u)$ for each $u$ simultaneously. 

\end{enumerate}

\begin{figure}
\centering
\begin{subfigure}[t]{0.47\textwidth}
\includegraphics[scale = 0.35]{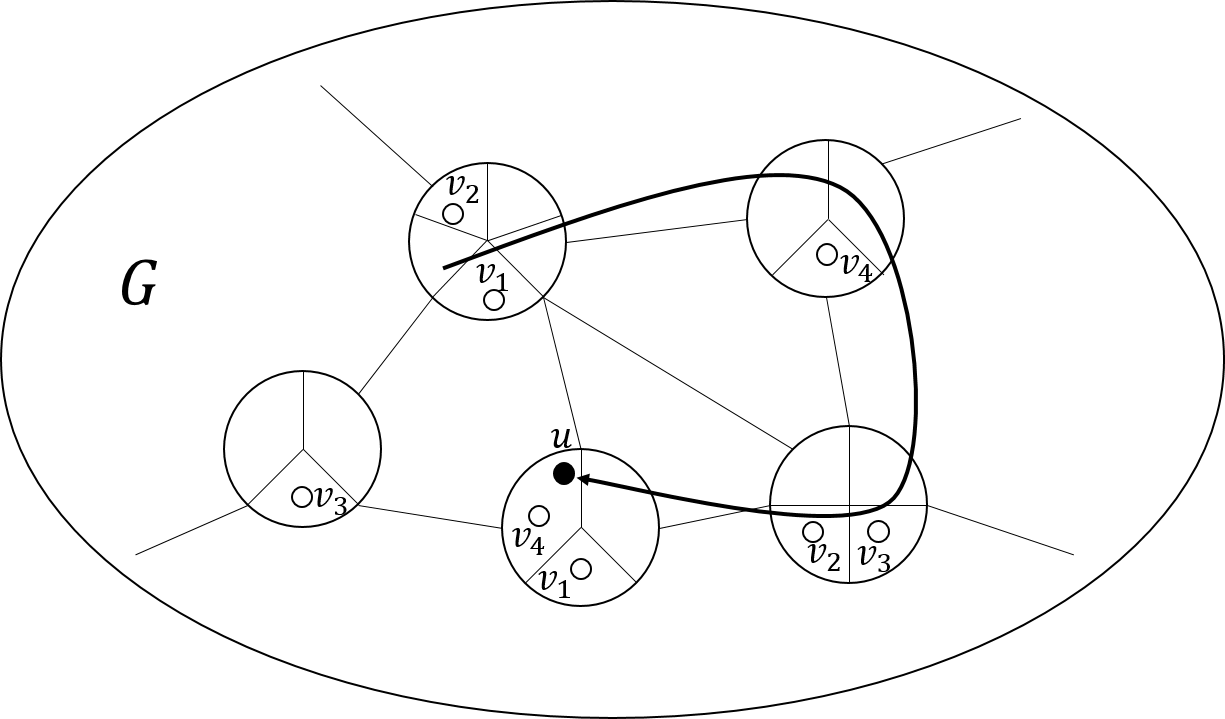} 
\caption{Illustration of Step \ref{step:2}. The random walk of $u$'s source token ends in the compartment containing the destination token of $v_4$. Thus, $t(u) = v_4$.}\label{fig:1a}
\end{subfigure}
\quad
\begin{subfigure}[t]{0.46\textwidth}\
\includegraphics[scale = 0.35]{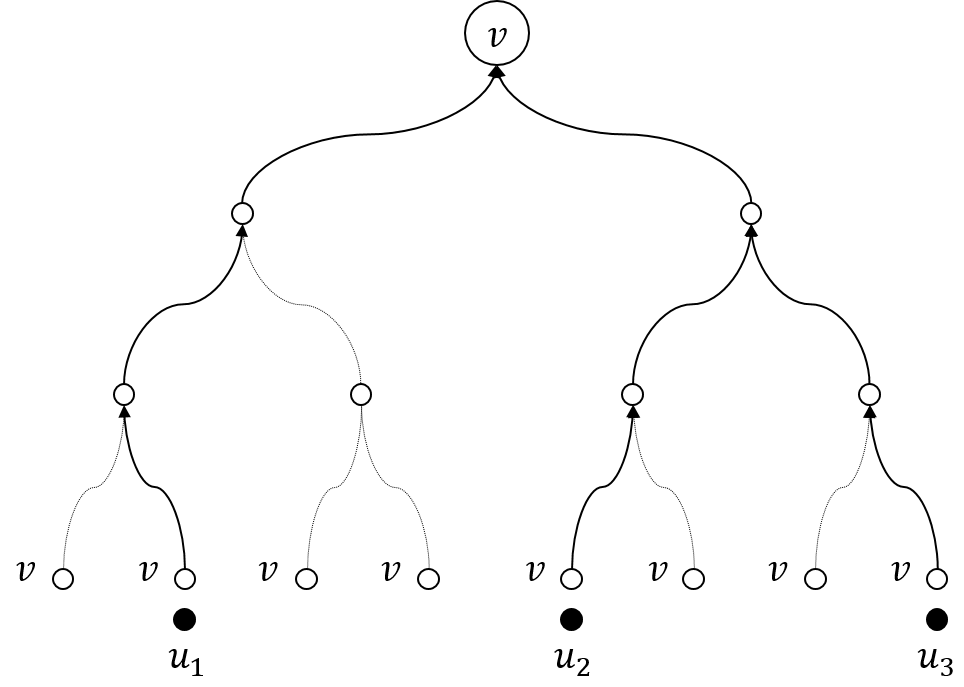}
\caption{Illustration of Step \ref{step:3}. $u_1,u_2,u_3$ will follow the paths taken by the destination tokens of $v$ back to $v$. The paths may overlap. W.h.p.~every edge is contained in at most $O(\log n)$  paths.}\label{fig:1b}
\end{subfigure}
\caption{}
\end{figure}

We explain how to implement each step in details.

\paragraph{Step 1} Each node $u$ creates a destination token $(u,k)$ initially. The first component of the token is its identity while the second component of the token is its multiplicity. The goal is to split the tokens and distribute them across the compartments so that all tokens have multiplicity of 1 and each compartment holds at most one token. We divide Step 1 into the {\bf splitting phase} and the {\bf distributing phase}. 

The splitting phase is further divided into $\lceil \log k \rceil$ stages. At the beginning of each stage,  if $W > 1$, each token $(u,W)$ is split into two tokens $(u,\lceil W/2 \rceil)$ and $(u,\lfloor W/2 \rfloor)$. Then all tokens perform $\tau_G$ steps of random walks.

We show that w.h.p., there are at most $O(\log n)$ tokens per compartment at the end of each stage. Given a stage, the probability that a token ends up in a given compartment in node $v$ is at most $$\left(\frac{\deg(v)}{2m} + \frac{1}{2mn}\right) \cdot \frac{1}{\deg(v)} \leq \frac{1}{m}.$$
Since there are at most $k \cdot n \leq 1.5m$ tokens, there are at most $O(1)$ tokens ending in a compartment in expectation. By standard Chernoff and union bound argument, w.h.p.~there are at most $O(\log n)$ tokens in each compartment.

Moreover, since each node $u$ holds at most $\deg(u) \cdot O(\log n)$ tokens at the beginning of each stage, the random walks can be performed in parallel in $O(\tau_G \cdot \log n)$ rounds by Lemma \ref{lem:parallelrandomwalk}. Therefore, the splitting phase uses $O((\log k) \cdot (\tau_G \cdot (\log n)) ) = \tilde{O}(\tau_G)$ rounds. At the end of the splitting phase, the multiplicity of each token is one. Moreover, w.h.p.~each compartment contains at most $O(\log n)$ tokens.

In the distributing phase, a compartment containing more than one token will start the random walks on {\it all except one} of its token for $\tau_G(0.1/2m)$ steps. Again, by Lemma \ref{lem:parallelrandomwalk}, this can be done simultaneously for all nodes in $O(\tau_G \cdot \log n)$ rounds. At the end of the random walks, we say a token succeeds if it ends at a compartment without any other tokens. If a token does not succeed, it will go back to the origin. The process is repeated until there is no compartment containing more than one token. Since there are at most $n \cdot k \leq 1.5m$ tokens, at most $1.5m$ compartment can be occupied.  Since we run the random walks for $\tau_G(0.1/2m)$ steps, the probability that a random walk ends at a specific compartment is at most $1.1/2m$. Thus, the probability that a token does not succeed is at most $(1.5m) \cdot (1.1/ (2m)) = 1.65/2 $.

Therefore, a token will succeed w.h.p.~after at most $O(\log n)$ trials. By a union bound over the tokens, w.h.p.~all tokens succeed after $O(\log n)$ trials. The total running time is $O(\log n \cdot (\tau_G \log n)) = \tilde{O}(\tau_G)$.

\paragraph{Step 2} Each node $u$ creates a source token. The tokens start to perform random walk for $\tau_G(\lambda')$ steps, where $\lambda' = \min(\lambda/(8m), 0.1/m )$ (see Figure \ref{fig:1a}). If the source token of $u$ ends up in one out of the $k$ compartments with a destination token of $v$, $t(u)$ will be set to $v$. Otherwise, if it ends up in a compartment without any destination tokens, it will restart the random walk. The process will be repeated until the source token ends up in a compartment with some destination token. 

By our choice of $\lambda'$, the probability that a token ends at a specific node is at least $0.9/(2m)$. Therefore, the probability that a token successfully ends up in a compartment with a destination token after the random walk is at least 
\begin{align*}
n   k  \cdot \frac{0.9}{2m} &\geq (1.5m - n) \cdot \frac{0.9}{2m} \geq (1.5m - m-1) \cdot \frac{0.9}{2m}  
\geq \left(\frac{1}{4} - \frac{1}{2m} \right) \cdot 0.9  \geq 0.9/8 ~. 
\end{align*}

The second inequality follows from $ m \geq n-1$ and the third inequality holds for $m \geq 4$. Thus, the number of random walks a token needs to perform until it ends up at a node with some destination token is at most $O(\log n)$ w.h.p. By taking a union bound over all the $n$ tokens, we conclude that w.h.p.~every token performs at most $O(\log n)$ random walks. The random walks can be performed simultaneously in $O(\tau_G \cdot \log n)$ rounds, so w.h.p.~the total number of rounds is $O(\tau_G \cdot \log^2 n)$.

Next, we show that given two nodes $u,v$, $\Pr(t(u)= v) \in [(1-\lambda)/n, (1+\lambda)/n]$. Let $\mathcal{E}_v$ denote the event that the source token of $u$ ends up in a compartment with a destination token of $v$. Let $\mathcal{E}$ denote the event that the source token of $u$ ends up in a compartment with some destination token. 

By our choice of $\tau(\lambda')$, we have that for all $v$, $\Pr(\mathcal{E}_v) \in \left[\frac{k}{2m}-k\lambda', \frac{k}{2m}+k\lambda' \right]$ and $\Pr(\mathcal{E}) \in \left[n \left(\frac{k}{2m}-k\lambda' \right), n \left(\frac{k}{2m}+k\lambda'\right) \right]$. Therefore,
\begin{align*}
\frac{\frac{k}{2m} - k\lambda'}{n \left(\frac{k}{2m} + k \lambda' \right) }	&\leq \Pr(t(u) = v) \leq \frac{\frac{k}{2m} + k\lambda'}{n \left(\frac{k}{2m} - k \lambda' \right)} \\
\frac{1}{n} \cdot \frac{1 - 2m\lambda'}{1 + 2m\lambda'}	&\leq \Pr(t(u) = v) \leq \frac{1}{n} \cdot \frac{1 + 2m\lambda'}{1 - 2m\lambda'} \\
\frac{1}{n} \cdot (1- 8m\lambda')&\leq \Pr(t(u) = v) \leq \frac{1}{n} \cdot (1+8m\lambda') && \mbox{when $\lambda'$ is sufficiently small} \\
\frac{1}{n} \cdot (1- \lambda)&\leq \Pr(t(u) = v) \leq \frac{1}{n} \cdot (1+ \lambda) && \lambda' \leq \lambda / 8m~.
\end{align*}

Note that since all the source tokens perform random walks independently, when we condition on the choice of nodes in $Z$ for any $u \notin Z \subseteq V$, it is still true that 
\[
\Pr \left(\mathcal{E}_v \big| \bigwedge_{z \in Z} t(z) \right) \in \left[\frac{k}{2m}-k\lambda', \frac{k}{2m}+k\lambda' \right]\] 
and 
\[
\Pr \left(\mathcal{E} \big| \bigwedge_{z \in Z} t(z) \right) \in \left[n(\frac{k}{2m}-k\lambda'), n(\frac{k}{2m}+k\lambda')\right]~.\] 
Thus, $\Pr \left(t(u)= v | \bigwedge_{z \in Z} t(z) \right) \in [(1-\lambda)/n, (1+\lambda)/n]$.

\paragraph{Step 3} It remains to show that the messages from $u$ to $t(u)$ can be routed simultaneously for every $u$ in $\tilde{O}(\tau_G)$ rounds. 

Let $mid(u)$ denote the node where the source token of $u$ is located at the end of Step \ref{step:2}. The message from $u$ to $mid(u)$ for every $u$ can be simultaneously routed in $\tilde{O}(\tau_G)$ rounds by following the same path taken by the random walk of the source token of $u$.

Suppose that $t(u) = v$. After the message reaches $mid(u)$, it will follow the path taken by the random walk of the destination token of $v$ to go to $v$ (see Figure \ref{fig:1b}). Note that multiple source tokens may be matched to a node $v$ (some possibly from the other destination tokens of $v$). When they follow the paths that lead back to $t(v)$, it is possible that these paths merge and create congestion. However, using a standard Chernoff Bound argument, we can show that for any node $v$ w.h.p.~at most $O(\log n)$ different nodes $u$ have $t(u) = v$. Therefore, each step of the parallel random walk can be done with a $O(\log n)$ factor blowup. Thus, the messages  between $u$ and $t(u)$ can be routed in $\tilde{O}(\tau_G)$ rounds. This completes the proof of Theorem \ref{thm:gossipemu}.

%!TEX root = main.tex

\section{Algorithms in the GOSSIP Model}\label{sec:gossip}

In this section, we show that if we have a small number of non-empty nodes, then the empty nodes help approximate $F_p$ faster.  As stated in Corollary \ref{cor:main4}, this result can be translated back to the $\congest$ model using Theorem \ref{thm:gossipemu} with a  blow-up factor $\tilde{O}(\mixingG)$. We exhibit a pre-processing step that duplicates the values so that $\Omega(n)$ nodes become non-empty which is crucial for the algorithms to work while preserving the occurrence ratios. 

Throughout this section, for the sake of clarity, we consider the $\gossip(0)$ model. However, running our algorithms in $\gossip(1/n^c)$, for some sufficiently large constant $c$, only incurs a small additive error $1/\poly(n)$.

\begin{lemma}\label{thm:gossip-preprocessing}
If the number of non-empty nodes $z < n/3$, we can duplicate the values  so that $z \roundup{(n/3)/z}$ nodes become non-empty while preserving the occurrences ratios in $O(\log^2 n)$ rounds in the \gossip model.
\end{lemma}

\begin{proof}%[Proof of Lemma \ref{thm:gossip-preprocessing}]

We divide the process into three phases. \paragraph{Pre-processing} We assume that the number of non-empty nodes is less than $n/3$, otherwise, we are done. First, the nodes compute the number of non-empty nodes $z$ in $O(\log n)$ rounds \cite{kempe2003gossip}. Each node $v$ will form a token that contains $\val{v}$ and $t$ where $t$ is originally set to $\roundup{(n/3)/z}$.

{\bf Splitting Phase:} This phase consists of $O(\log n)$ stages each of which consists of $O(\log n)$ sub-stages. At the beginning of each stage, a node $v$ has a collection of tokens $(x_1,t_1),(x_2,t_2),\ldots$ in its buffer. Each token $(x_i,t_i)$ is split into two tokens  $(x_i,\roundup{t_i/2})$ and $(x,\floor{t_i/2})$. It  will send these two tokens to two random nodes using two rounds and delete $(x_i,t_i)$ from its buffer. Note that the new tokens $(x_i,\roundup{t_i/2})$ and $(x,\floor{t_i/2})$ will not be split until the next stage. Every stage produces at most $z\roundup{(n/3)/z} \leq 2n/3$ new tokens. Each new token is sent to a random node and therefore each node contains $O(\log n)$ new tokens w.h.p by Chernoff bound at the end of that stage. Hence, each sub-stage requires at most $O(\log n)$ rounds to split all the tokens in its buffer w.h.p. After $O(\log n)$ stages,   w.h.p  all nodes contain $O(\log n)$ tokens and  all tokens $(x,t)$ satisfy $t = 1$.

{\bf Distributing Phase:} At this point, we  only have tokens in the form $(x,1)$, or simply $x$. In each stage, if $v$ holds more than one token, it will send all but one  token (say the first that arrives at $v$) to the nodes that it talks to. By a standard Chernoff bound argument, each stage requires $O(\log n)$ rounds since each node always holds at most $O(\log n)$ tokens w.h.p. 
We say a token $x$ succeeds if it lands in a previously empty node $u$ while no other token lands in $u$ in the same round. Then, $u$ never sends $x$ away from this point onward. 
Since we have at most $z \cdot \roundup{(n/3)/z} \leq 2n/3$ tokens, at least $n/3$ nodes are empty at all times. Consider a  token $x$. In each stage, conditioning on all other tokens' choices, with probability at least $1/3$,  $x$ succeeds. Hence, after $O(\log n)$ stages, $x$ succeeds w.h.p and therefore all tokens succeed w.h.p by taking a union bound over all tokens. Since we have at least $\roundup{n/3}$ tokens, the number of non-empty nodes is $\Omega(n)$. Note that the occurrence of each value is rescaled by a factor $\roundup{(n/3)/z}$.
\end{proof}

After we estimate $F_p$ of the new instance, we can divide the estimator by $(\roundup{(n/3)/z})^p$ to get an estimate for $F_p$ in the original instance. From now on, we can safely assume that the number of non-empty nodes $F_1 = \Omega(n)$, otherwise, we can apply the above pre-processing. A {\em key observation} is that  $F_0 \leq z$, and thus we can analyze our algorithms for when $F_0$ is small instead. 

\paragraph{An $\ell_p$-sampling primitive} An  $\ell_p$-sampling algorithm samples a value $i \in [N]$ with probability $f_i^p / F_p$. More formally, $\Prob{\text{sample $i$}} = {f_i^p}/{F_p}$.

The $\ell_p$-sampling primitive (for $0 \leq p \leq 2$) has been extensively studied in the data stream model. An incomplete list includes \cite{AndoniKO11,JayaramW18,MonemizadehW10,JowhariST11}.  However, most streaming $\ell_p$-samplers are  rather complicated, and it is unclear how to implement them in the $\gossip$ model.

It is trivial to obtain an $\ell_1$-sample by virtue of the $\gossip$ model. To obtain an $\ell_0$-sample (a random value that occurs at least once), we broadcast a randomly chosen pairwise hash function $h: [N] \rightarrow [N^3]$ and identify the value corresponds to the smallest hash value in $O(\log n)$ rounds.

Assuming that $p$ is fixed, we now show that if $F_0 = O \left( n^{1/(p-1)} \right)$, then we can perform $\ell_p$-sampling in $O(\log n)$ rounds (hence $\ell_2$-sampling can always be done in $O(\log n)$ rounds since $F_0 \leq n$).  The sampling algorithm proceeds as follows. 

Each node $v$ uses $p$ rounds to talk to $p$ random nodes $u_1,\ldots,u_p$. It declares success if $\val{u_1} = \ldots = \val{u_p} $. In that case, let $\val{u_1}$ be $v$'s sample. Among the successful nodes, to break symmetry, broadcast the sample of the node with the smallest $\ID$. If no node succeeds, repeat the process. The following lemma provides a lower bound on $F_p$ based on $F_0$.
\begin{lemma} \label{lem:F0-Fp}
If $F_1= \Omega(n)$ and $F_0 = O \left(n^{1/(p-1)} \right)$, then $F_p  = \Omega \left( n^{p-1} \right)$.
\end{lemma}

\begin{proof}%[Proof of Lemma \ref{lem:F0-Fp}]
Let the frequency vector be $f = \left( f_1,\ldots,f_N \right)$. Without loss of generality, suppose the potentially non-zero entries of $f$ be $f_1,\ldots,f_{K n^{1/(p-1)}}$ for some constant $K$. Note that based on our assumption, $f_j = 0$ for all $j > K n^{1/(p-1)}$. Let $f' = \left(f_1,\ldots, f_{K n^{1/(p-1)}} \right)$ be the vector formed by the first $K n^{1/(p-1)}$ entries.  Note that $\| f' \|_1 \geq C n$ for some constant $0 < C \leq 1$ as assumed.

We will use the following inequality: if the vector $x$ has $n$ entries then
\[
\| x \|_q \leq \left(n^{1/q - 1/p} \right) \| x  \|_p~ \text{, for $0 < q < p$~.}
\]

Note that $f'$ has $K n^{1/(p-1)} $ entries. Let $K' = K^{1-1/p}$. We have 
\begin{align*}
\left(K n^{1/(p-1)} \right)^{1-1/p } \left\| f' \right\|_p & \geq  \| f' \|_1 \\
\| f' \|_p & \geq \frac{\| f' \|_1}{K' n^{1/p}} \\
F_p & \geq \frac{C^p n^p}{K^{p-1} n} = \Omega\left(n^{p-1}\right)~.% \geq \Omega(n^{p-1})~.
\end{align*}
The last step follows since $K$ and $C$ are constants and $p$ is fixed.
\end{proof}

\begin{theorem}
If $F_0=O\left(n^{1/(p-1)}\right)$, then the described algorithm obtains an  $\ell_p$-sample in $O(\log n)$ rounds in the \gossip model w.h.p.
\end{theorem}
\begin{proof}
We can apply the pre-processing step so that $F_1 = \Omega(n)$ while the occurrences ratios are preserved. The probability that a node succeeds is $\Omega \left(\sum_{i=1}^N {f_i^p}/{n^p} \right) = \Omega \left({F_p}/{n^p} \right)$.

 Appealing to Lemma \ref{lem:F0-Fp}, $F_p \geq n^{p-1}/K'$ for some constant $K'$. Hence, $\Prob{\text{$v$ succeeds}} \geq 1/(K' n)$. The probability that all $n$ nodes fail is at most $\left( 1-1/(K' n) \right)^n \leq e^{-1/K'}$. We therefore succeed w.h.p   by repeating $O(\log n)$ times. Given that $v$ succeeds, the probability that it samples value $i$ is $ \left({f_i^p/n^p} \right)/ \left({\sum_{j=1}^N f_j^p/n^p}\right) = {f_i^p}/{F_p}$ as required.
\end{proof}
\paragraph{Approximating $F_p$} The algorithm by Bar-Yossef et al. \cite{Bar-YossefJKST02} that we discuss in Appendix \ref{appendix:adaptation} for approximating $F_0$ up to a $1\pm \eps$ factor w.h.p  can be emulated in the $\gossip$ model in $O(\eps^{-2} \log^2 n)$ rounds. We now focus on approximating higher frequency moments. Let $k \leq p$ be an integer. We present an algorithm that w.h.p  approximates $F_p$ (for $p \geq 2$) in $\tilde{O}\left(\eps^{-2} n^{1-k/p}\right)$ rounds if $F_0 = O\left(n^{1/(k-1)} \right)$. Recall that $F_0$ is at most the number of non-empty nodes. To approximate $F_p$, our algorithm makes use of an approximation of $F_k$ and $\ell_k$-sampling. This generalizes the approach in \cite{AndoniKO11,MonemizadehW10}. We will prove the following theorem.

We first consider the following algorithm that approximates $F_k$. For $j=1,\ldots, C \eps^{-2} \log n$, where $C$ is some sufficiently large constant, in the $j$-th phase, each non-empty node $v$ uses $k-1$ rounds to talk to $k-1$ random nodes $u_1,\ldots,u_{k-1}$. It declares success if $\val{v} = \val{u_1} = \ldots = \val{u_{k-1}}$. Let $I_{j,v}$ be the indicator variable for the event $v$ succeeds in the $j$-th phase.
 Let $T = C \eps^{-2} \log n$. Return the estimate 
 \[
 \hat{F}_k = \frac{n^{k-1}}{T} \cdot \sum_{j=1}^T \sum_{v=1}^n I_{j,v}~.
 \]

We now prove Theorem \ref{thm:fp-gossip}. This theorem first shows that $\hat{F}_k$ is a good approximation w.h.p. Then, it combines  $\hat{F}_k$ with $\ell_k$-sampling to compute a good estimate of $F_p$ in $O\left(\eps^{-2} n^{1-k/p}\log^2 n \right)$ rounds.

\begin{proof}[Proof of Theorem \ref{thm:fp-gossip}]
We again can assume that $F_1 = z = \Omega(n)$ as outlined earlier in this section. We first show that $\hat{F}_k = (1\pm \eps)F_k$ w.h.p. In expectation,
\begin{align*}
\Expect{\hat{F}_k} & = \frac{n^{k-1}}{T} \sum_{j = 1}^T \sum_{v=1}^n \Expect{I_{j,v}}  = \frac{n^{k-1}}{T} \sum_{j = 1}^T \sum_{v=1}^n \frac{f_{\val{v}}^{k-1}}{n^{k-1}}  = \sum_{i=1}^N{f_i \cdot f_i^{k-1}} = F_k~.
\end{align*}
Since the indicator variables $I_{j,v}$ are independent, we can apply Chernoff bound directly.
\begin{align*}
\Prob{\left| \hat{F}_k - F_k \right| \geq \eps F_k} & = \exp\left( -\Omega\left( \frac{T\eps^{2}F_k}{n^{k-1}} \right) \right)  \leq  \exp\left( -\Omega\left( T\eps^{2}\right) \right) \leq 1/\poly(n)~.
\end{align*}
The first inequality follows from Lemma \ref{lem:F0-Fp}  and the second inequality is because $T = C \eps^{-2}\log n$ for some sufficiently large constant $C$. Hence, we can approximate $F_k$ up to a $1 \pm \eps$ factor in $O(\eps^{-2} \log n)$ rounds.

To approximate $F_p$ for $p > k$, we use the following estimator. Let $i$ be an $\ell_k$ sample. We can compute $f_i$ exactly  in $O(\log n)$ rounds. Specifically, each node with value $i$ will put  $1$ on it and 0 otherwise. Then, we can compute the sum using the algorithm in \cite{kempe2003gossip}.  Consider the following estimator:
\[
\hat{F}_p = \hat{F_k} \cdot f_i^{p-k}~.
\]
We rely on the following lemma. We defer the proof to the end of this section.
\begin{lemma}\label{lem:fp-estimator}
We have $\Expect{\hat{F}_p} = (1 \pm O(\eps)) F_p$ and $\Var{\hat{F}_p}  \leq 2 n^{1- k/p} F_p^2$.
\end{lemma}

 Hence, by an application of Chebyshev  bound, if we take the average of  $O\left( n^{1- k/p} \eps^{-2} \right)$ estimators, with constant probability, $\hat{F_p} = (1 \pm \eps)F_p$. We can amplify the success probability to $1-1/\poly(n)$ by the standard median trick, i.e., taking the median of $O(\log n)$ such estimators. 
\end{proof}

\begin{proof}[Proof of Lemma \ref{lem:fp-estimator}]
In expectation,
\begin{align*}
\Expect{\hat{F}_p} & = \hat{F_k} \cdot \sum_{i=1}^N \frac{f_i^k}{F_k} {f}_i^{p-k} \\
& = (1 \pm O(\eps)) F_p~.
\end{align*}
We can bound the variance as follows.
\begin{align*}
\Var{\hat{F}_p} & \leq (1\pm O(\eps))  F_k^2 \sum_{i=1}^N \frac{f_i^k}{F_k} {f}_i^{2(p-k)} \\
& \leq 2 F_k F_{2p-k}~.
\end{align*}
We have $\| f \|_k  \leq n^{1/k-1/p} \| f \|_p$, and therefore $F_k  \leq n^{1- k/p} F_p^{k/p}$. Additionally, $\| f \|_{2p-k} \leq \| f \|_p$ which implies $F_{2p-k} \leq F_{p}^{2-k/p}$. Therefore, $\Var{\hat{F}_p}  \leq 2 n^{1- k/p} F_p^2$.
\end{proof}

\bibliographystyle{alpha}
\bibliography{ref}

%!TEX root = main.tex

\appendix

%\section{Algorithmic Primitives}\label{appendix:primitives}

\section{Upper and Lower Bounds for Approximating $F_p$ in the CONGEST model} \label{appendix:adaptation}
\subsection{Upper bounds}
\paragraph{Approximating $F_0$} We show how to run the algorithm by Bar-Yossef et al. \cite{Bar-YossefJKST02} in the $\congest$ model.  Kuhn et al. \cite{KuhnLS08} also outlined how to run  a different streaming algorithm for estimating $F_0$. However, the number of rounds for a $1 \pm \eps$ approximation w.h.p is unclear in their  paper.

The algorithm is as follows. We pick a pairwise hash function $h: [N]\rightarrow [M]$. Let $t=\roundup{100\eps^{-2}}$, $M = N^3$, and $w$ be the $\roundup{t}$-th smallest value among the hash values $W = \{ h(\val{v}) : v \in V \}$. We have the following:   $tM/w$ is a $1\pm \eps$ approximation of $F_0$ with probability at least 2/3 \cite{Bar-YossefJKST02}.

In the $\congest$ model, the leader can broadcast the hash function. Then, each node $v$ computes $h(\val{v})$. As mentioned above, the nodes can find the $t$-th smallest hash value in $W$ in $O(D+t)$ rounds. Hence, the total number of rounds is $O(D + \eps^{-2})$ for constant success probability. 

To amplify the success probability to $1-1/\poly(n)$, we use $\log n$ different hash functions and take the median of the corresponding estimates. At first, it is unclear how to pipeline this approach on the BFS tree to run in $O(D + \eps^{-2} \log n)$ rounds instead of $O(\log n \cdot (D + \eps^{-2}))$ rounds. We show that this is possible.  The following lemma generalizes Lemma 4.3.1 in  \cite{Peleg:2000}.

\begin{lemma} \label{lem:upcasts}
Suppose each element  belongs to exactly one of $k$ groups. Upcasting the $t$ smallest elements of $k$ groups on a tree $T$  can be performed in $O(\depth(T) + kt)$ rounds in the \congest model.
\end{lemma}

\begin{proof} %[Proof of Lemma \ref{lem:upcasts}]
Let the levels of the nodes be $1,2,\ldots, \depth(T)$ where the root is at level $\depth(T)$. For each node $v$, from round $\textup{level}(v)+ j$ to $\textup{level}(v) + j + t-1$, it sends the smallest value in group $j$ in its memory that has not been sent before to its parent. 

Suppose  that the $i$-th smallest element of group $j$ is in the subtree $T_v$ (the subtree that is rooted at $v$) originally. Then, we claim that at the end of round $\textup{level}(v) + j + i-1$, it will be stored in $v$. Furthermore, at the end of round $\textup{level}(v) + j + i$, it will be upcasted to $v$'s parent. We prove by triple-induction on $\textup{level}(v), j$ and $i$. The base case where $\textup{level}(v)=1, j=1, i=1$ can easily be checked.

Consider a node $v$ where $\textup{level}(v) = \ell$. Suppose the $i$-th smallest element $x$ of group $j$ is in $T_v$ originally. We first need to show that at the end of round $\ell + j +i - 1$, the element $x$ is sent to $v$. We know that one of $v$'s children, say $u$, must have $x$ in $T_u$; note that $\textup{level}(u) = \ell-1$. By induction, at the end of round $(\ell-1) + j + i$,  we know that $u$ must have sent $x$ to $v$.  It remains to show that at the end of round $\ell + j + i$, $v$ will send $x$ to its parent. By induction, all the $i'$-th smallest elements of group $j$ that are in $T_v$, where $i' < i$, must have been upcasted to $v$'s parent at the end of round $\ell + j + i - 1$. Hence, $v$ must  upcast the $i$-th smallest value of group $j$ to its parent at the end of round $\ell + j + i$. Therefore, after $O(\depth(T) + kt)$ rounds, the root must have all the desired values. It can downcast them back to other nodes in $O(\depth(T) + kt)$ rounds.
\end{proof}

As a result, we can find all $O(\eps^{-2})$ smallest hash values of each of $O(\log n)$ hash functions in $O(D + \eps^{-2}\log n)$ rounds using the BFS tree.

\paragraph{Approximating $F_2$}  To approximate the second frequency moment $F_2$, we adapt the well-known tug-of-war sketch in  \cite{AlonMS99} to the $\congest$ model. The adaptation, using $O(D + \eps^{-2}\log n)$ rounds, for a $1 \pm\eps$ approximation with high probability is quite simple and can be found in \cite{KuhnLS08}.

\paragraph{Approximating $F_p$ (for $p > 2$)} The AMS sketch for $F_p$ in \cite{AlonMS99} can be adapted to the $\congest$ model as follows. First, the leader samples a node $v$ and try to compute $r = \size{\{ u  \in V : u > v \land \val{u} = \val{v} \}}$. We can compute $r$  in $O(D)$ rounds. After the leader broadcasts $v$, we let $v$ broadcasts its value $\val{v}$. Finally, each node $u$ now knows if $u > v$ and $\val{u} = \val{v}$.  Then, we can compute the sum $r = \sum_{u \in V} \indicator{u > v \land \val{u} = \val{v}} $. All of these steps can be done in $O(D)$ rounds. Let $X = n (r^p - (r-1)^p)$. One can argue that $\Expect{X} = F_p$. Furthermore, if we repeat $O( \eps^{-2} \min(n,N)^{1-1/p} \log n)$ times and take the average as the final estimate $\hat{X}$, it can be shown that $\hat{X}$ is a $1\pm \eps$ approximation of $F_p$ with high probability \cite{AlonMS99}.Thus, we obtain an $O(D + \eps^{-2} \min(n,N)^{1-1/p} \log n)$-round algorithm via a careful pipelining.

% If $\eps^{-2} N^{1-1/k} \gg n$, we can instead let the leader collect all the values in the graph in $O(n)$ rounds. 

We summarize the discussion above in Theorem \ref{thm:approximate-congest-f02}.

\subsection{Lower bounds}
\paragraph{$\Omega(D)$ lower bound for approximating $F_p$ up to a $1\pm 0.1$ factor} We now  present lower bounds suggesting that these algorithms are tight. This provides  a clear context and a good motivation for our exact algorithm in Section \ref{sec:exact-algorithm} which is one of our main results.  It is easy to see that  $\Omega(D)$ rounds are required for some constant approximation. Consider $n$ nodes $a_1,\ldots,a_n$ that are connected to one end $c_1$ of a chain $c_1,\ldots,c_D$.  Let the nodes $b_1,\ldots,b_n$ connect to the other end $c_D$ of the chain. Set $\val{a_i} = i$ for all $i$ and $\val{c_j} = \textup{NULL}$ for all $j$. In the first case, $\val{b_i} = i$ and in the second case $\val{b_i} = i + n$. In order for each $a_i$ to distinguish between the two cases, $\Omega(D)$ rounds are clearly needed. In the first case, $F_0 = n$ and in the second case $F_0 = 2n$ (for $p \neq 0$, the first case  corresponds to $F_p = n 2^p$ and the second case corresponds to $F_p = 2n$). A $(1\pm 0.1)$ approximation for $p \in \{0,2,3,\ldots\}$ distinguishes the two cases and therefore requires $\Omega(D)$ rounds.

\paragraph{Conditional $\tilde{\Omega}(\eps^{-2})$ lower bound for approximating $F_p$ up to a $1 \pm \eps$ factor} We now reason why the above algorithms might be optimal for a $1 \pm \eps$ approximation, in terms of $\eps$, based on a conjecture of Brody-Chakrabarti \cite{BrodyC09}. For $F_p$ (where $p \neq 1$ is a fixed), we consider the following communication problem in which Alice and Bob have the sets $A \subseteq [N]$ and $B \subseteq [N]$ respectively. Let $C$ be the multiset formed by $A$ and $B$. The goal is to estimate $F_p(C)$ up to a $1 \pm 1/\sqrt{N}$ factor with some sufficiently large constant success probability. For this problem, Woodruff gave a one-round lower bound $\Omega(N)$ on the number of bits that Alice and Bob need to communicate,  via a reduction from the one-way Gap-Hamming-Distance problem of size $N$ \cite{Woodruff04}. Hence,  if $\eps=\Theta(1/\sqrt{N})$, this implies an $\Omega(\eps^{-2})$ lower bound. 

It is conjectured in \cite{BrodyC09} (see conjecture 2) that the total communication of the Gap-Hamming-Distance problem of size $N$, irrespective of the number of rounds,  must be $\Omega(N)$ bits. Hence, if Alice and Bob can communicate $O(\log n)$ bits in each round, the number of rounds must be $\Omega(\eps^{-2}/\log n)$.  

Assuming their conjecture holds, then in the worst case, $\Omega(\eps^{-2}/\log n)$ rounds are needed in the $\congest$ model to approximate $F_p$ up to a $1 \pm \eps$ factor. To see this, consider a graph $G$ of $2N$ nodes, i.e., $n = \Theta(N)$. If $i \in A$, then Alice sets $\val{i} = i$, otherwise, she sets $\val{i} = \textup{NULL}$. Similarly, if $i \in B$ then Bob sets $\val{N+i} = i$, otherwise, he sets $\val{N+i} = \textup{NULL}$. One example that works is that $L$ and $R$ are cliques and $(1,N+1) \in E$. In fact, the conditional lower bound holds for any graph in which $L$ and $R$  are connected component and $|E(L,R)|=O(1)$.

Observe that $G$ is connected as required, but in each round, only $O(\log n)$ bits can be communicated between $L$ and $R$ via the edge $(1,N+1)$. Hence, if we can approximate $F_p(C)$ up to a $1\pm \eps$ factor in $G$ in $r$ rounds, it means that  Alice and Bob can use this as a protocol to approximate $F_p(C)$  using $O(r \log n)$ bits of communication and  therefore $r = \Omega(\eps^{-2}/\log n)$. In this construction, $n = \Omega(N)$.%Thus, if $\eps = 1/\sqrt{n}$, this implies an $\Omega(\eps^{-2}/\log n)$ rounds (conditional) lower bound.

\paragraph{$\tilde{\Omega}(N^{1-2/p})$ lower bound for approximating $F_p$ up to a $1\pm 0.1$ factor} We also observe that approximating $F_p$ (for $p > 2$) cannot be done in fewer than $O(N^{1-2/p}/\log n)$ rounds.  This is based on a suitable modification of the reduction in \cite{AlonMS99}. In the $t$-player disjointness problem, player $i$ has the set $A_i$ and the players want to learn (with some sufficiently high constant success probability) if the all the sets are disjoint (YES case) or they intersect at a unique element (NO case), with the promise that one of the two cases happens. In the blackboard model, the players can send messages to a blackboard for others to see. The total size of all the messages must be at least $\Omega(N/t + \log N)$ \cite{ChakrabartiKS03}. Consider a graph with $t$ parts each of which is a clique of $N$ nodes. The $i$-th part encodes $A_i$ in the same fashion above. It is easy to see that if $t > 2^{1/p} N^{1/p}$, then, a $(1 \pm 0.1)$-approximation of $F_p$ can distinguish the YES case ($F_p \leq N$) and the NO case ($F_p \geq t^p > 2 N$). Each part connects with a node $b$ that serves as a blackboard via  a single edge. In each round in the $\congest$ model, the players can send $O(t \log n)$ bits of message to the blackboard in total. Hence, the number of rounds must be $\Omega(N/(t^2 \log n)) = \Omega(N^{1-2/p}/\log n) = \Omega(n^{(1-2/p)/(1+1/p)}/\log n)$, since $n = t \cdot N = \Theta(N^{1+1/p})$. %Note that this lower bound is unconditional. Note that in this set up, $n = Nt \geq 2^{1/p} \cdot N^{1+1/p}$. However, the number of empty node

\paragraph{An $\tilde{\Omega}(n)$ lower bound for computing $F_p$ exactly} Next, we give a simple unconditional lower bound for computing  $F_p$ exactly. Consider the 2-player disjointness problem. %: Alice and Bob want to distinguish between the two cases $|A \cap B| = 1$ (NO case) and $A \cap B = \emptyset$ (YES case).  As discussed, the communication complexity is $\Omega(N)$ regardless of the number of rounds. Now, we use the same construction as above.   
In the NO case, we have $F_0 = |A| + |B| -1$ and in the YES case, we have $F_0 = |A| + |B|$. For general $p$, the NO case corresponds to $F_p = |A| + |B| -2+2^p$ and the YES case corresponds to  $F_p = |A| + |B|$. Hence, for $p \neq 1$, given an exact algorithm for $F_p$ in the $\congest$ model that uses $r$ rounds, the two players can use that to solve Disjointness with  $O(r \log n)$ bits of communication (they use 2 extra rounds to send each other $|A|$ and $|B|$). Hence, $r = \Omega(n/\log n)$.

\end{document}